\documentclass{eptcs}
\providecommand{\event}{TERMGRAPH 2011} 

\providecommand{\firstpage}{1}

\usepackage{graphicx,epic}

\usepackage{nets2,pict2e}
\usepackage{livre,prooftree,tabularx,array}

\usepackage{mathrsfs}

\newcommand{\ra}{\rightarrow}
\newcommand{\Ra}{\Rightarrow}
\newcommand{\X}{{\cal X}}

\newcommand{\NLog}{\mathscr{N}_{\mathcal{L}}}

\newtheorem{proposition}{Proposition}
\newtheorem{lemma}{Lemma}
\newtheorem{definition}{Definition}
\newtheorem{example}{Example}
\newenvironment{proof}{\noindent{\bf Proof:}}{}

\def\qed{\pushright{$\square$ \vspace{2 mm}}}
\def\pushright#1{{
   \parfillskip=0pt            
   \widowpenalty=10000         
   \displaywidowpenalty=10000  
   \finalhyphendemerits=0      
   \leavevmode                 
   \unskip                     
   \nobreak                    
   \hfil                       
   \penalty50                  
   \hskip.2em                  
   \null                       
   \hfill                      
   {#1}                        
   \par}}                      

\title{A new graphical calculus of proofs}

\author{Sandra Alves
\institute{DCC - Faculty of Science \& LIACC\\ 
University of Porto\\
R. do Campo Alegre 1021/55, \\
4169-007, Porto, Portugal}
\and
Maribel Fern\'andez
\institute{Department of Informatics\\
King's College London\\
Strand, London, WC2R 2LS U.K}
\and
Ian Mackie
\institute{LIX, CNRS UMR 7161\\
\'Ecole Polytechnique\\
91128 Palaiseau Cedex, France}
}

\def\titlerunning{A new graphical calculus of proofs}
\def\authorrunning{Alves, Fern\'{a}ndez \& Mackie}
 
\begin{document}

\maketitle

\begin{abstract} 
  We offer a simple graphical representation for proofs of
  intuitionistic logic, which is inspired by proof nets and
  interaction nets (two formalisms originating in linear logic). This
  graphical calculus of proofs inherits good features from each, but
  is not constrained by them. By the Curry-Howard isomorphism, the
  representation applies equally to the lambda calculus, offering an
  alternative diagrammatic representation of functional computations.

\end{abstract}

\section{Introduction}\label{sec:intro}
There are many different ways to write proofs in a given logic, for
instance, natural deduction, sequent calculus and Hilbert systems are
well-known proof systems (we refer the reader
to~\cite{MackieIC:protad} for details).  Each syntax has advantages
and disadvantages. For example, classical logic works well in sequent
calculus because it allows the symmetry of the connectives to be seen;
a natural deduction presentation of classical logic is considered
artificial since rules are needed which do not correspond to the
introduction or elimination of a connective. Intuitionistic logic,
which links exceptionally well with computation, works better in
natural deduction, where proofs correspond to programs and there is a
notion of canonical proof (which is not the case in sequent calculus).

In this paper we focus on intuitionistic logic, and we aim at designing a
syntax that can both 
\begin{itemize}
\item
facilitate the visualisation and understanding of proofs, and 
\item
serve as a basis for the
implementation of the simplification rules in the logic. 
\end{itemize}
We choose intuitionistic logic as a basis for this work because of the
computational significance of the logic through the Curry-Howard
isomorphism: proofs in this logic correspond to functional programs;
logic formulas correspond to types of programs; and proof
normalisation corresponds to computation. Thus we are not looking for
a visual representation of \emph{truth}, but rather that of a
\emph{proof}. Furthermore, we are not interested
in just representing logics, but also in studying the reduction
process (normalisation) which corresponds to computation through the
Curry-Howard isomorphism.

Linear logic~\cite{GirardJY:linl} comes equipped with a graphical
syntax called proof nets. One of the motivations for the adoption of a
graphical syntax is that traditional syntaxes for logic, such as the
sequent calculus, have a lot of constraints to do with the
formalism---and not with the logic.  To illustrate this, we borrow a
well-known example from Girard~\cite{GirardJY:linlss}. Let $(r)$ and
$(s)$ be two logical rules, and consider a cut working on auxiliary
formulas (not the main formulas of the rules $r$ and $s$):
\[
\begin{prooftree}
\[
\vdash \Gamma , A
\justifies
\vdash \Gamma',\CF{A}
\using (r)
\]
\qquad
\[
\vdash \Aperp,\Delta
\justifies
\vdash \CF{\Aperp},\Delta'
\using (s)
\]
\justifies
\vdash \Gamma',\Delta'
\Cut
\end{prooftree}
\]
If we permute the cut rule up through $r$ and $s$, then there are two
possible choices depending on whether we first permute through $s$ or
through $r$. These two choices are represented by:
\[
\begin{prooftree}
\[\[
\vdash \Gamma,\CF{A} \qquad \vdash \CF{\Aperp},\Delta
\justifies
\vdash \Gamma,\Delta
\Cut
\]
\justifies
\vdash \Gamma',\Delta
\using (r)
\]
\justifies
\vdash \Gamma',\Delta'
\using (s)
\end{prooftree}
\qquad\qquad
\begin{prooftree}
\[\[
\vdash \Gamma,\CF{A} \qquad \vdash \CF{\Aperp},\Delta
\justifies
\vdash \Gamma,\Delta
\Cut
\]
\justifies
\vdash \Gamma,\Delta'
\using (s)
\]
\justifies
\vdash \Gamma',\Delta'
\using (r)
\end{prooftree}
\]
This demonstrates that permutation of rules is inherent in the syntax.
In addition, there is no canonical representation of proofs, even for
cut-free proofs. For example, consider the following two alternative
proofs (in classical logic), of the formula: $\Aperp,C,A \land \Bperp,
B\land \Cperp$, which differ only by the order in which 
axioms are combined:

\[
\begin{prooftree}
\prooftree
\justifies
\vdash \Aperp,A
\axiom
\endprooftree
\prooftree
\prooftree
\justifies
\vdash \Bperp,B
\axiom
\endprooftree
\prooftree
\justifies
\vdash \Cperp,C
\axiom
\endprooftree
\justifies
\vdash C,\Bperp,B\land \Cperp
\using (\land)
\endprooftree
\justifies
\vdash \Aperp,C,A\land \Bperp ,B\land \Cperp
\using (\land)
\end{prooftree}
\qquad
\qquad
\begin{prooftree}
\prooftree
\prooftree
\justifies
\vdash \Aperp,A
\axiom
\endprooftree
\prooftree
\justifies
\vdash \Bperp,B
\axiom
\endprooftree
\justifies
\vdash \Aperp,B,A\land \Bperp
\using (\land)
\endprooftree
\prooftree
\justifies
\vdash \Cperp,C
\axiom
\endprooftree
\justifies
\vdash \Aperp,C,A\land \Bperp ,B\land \Cperp
\using (\land)
\end{prooftree}
\]

A graphical syntax will free us from these inessential
permutations.  To illustrate the point, the previous two proofs are
given by the following proof net:
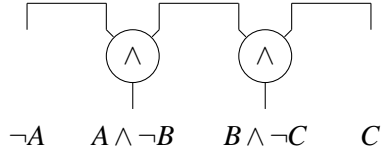
\begin{net}{150}{50}
\putHline{10}{50}{30}
\putHline{60}{50}{30}
\putHline{110}{50}{30}
\putVline{10}{40}{10}
\putVline{40}{40}{10}
\putVline{60}{40}{10}
\putVline{90}{40}{10}
\putVline{110}{40}{10}
\putVline{140}{40}{10}
\putland{50}{30}
\putland{100}{30}
\putline{40}{40}{1}{-1}{2.7}
\putline{90}{40}{1}{-1}{2.7}
\putline{60}{40}{-1}{-1}{2.7}
\putline{110}{40}{-1}{-1}{2.7}
\putVline{50}{10}{10}
\putVline{100}{10}{10}
\puttext{50}{0}{$A\land \Bperp$}
\puttext{100}{0}{$B\land \Cperp$}
\puttext{10}{0}{$\Aperp$}
\puttext{140}{0}{$C$}
\end{net}
In this visual representation, the nodes represent the connectives and
the edges are labelled by formulas.  Unlike the sequent calculus, we
do not need to impose an order on the introduction of
$\land$. Moreover, many of the permutation equivalences become
identities in proof nets.

In recent years, advances in this area have produced a better
understanding of the notion of proof, and a pleasing mathematical
theory has been developed. However, it is not clear that this is the
best way of visualising proofs or mechanising proof transformations.
Proof nets work very well for some fragments of the logic, but less
well for others where extra structure is required (boxes for
instance).  Related systems, such as interaction
nets~\cite{LafontY:intn} have been used to provide implementations of
proof nets. Indeed, they can be seen as a generalisation of proof
nets. However, the implementation nature of interaction nets often
leads to cluttering proofs with low-level details, which do not aid
the understanding of proofs.

In this paper we propose to take features from both proof nets and
interaction nets to build a hybrid system that takes some of the good
features of each, but is not limited by either of them. We put the
case forward for intuitionistic logic, and by the Curry-Howard
isomorphism, we also get results for the $\lambda$-calculus.  This
hybrid notation allows us to choose the level of detail that we want
to include in the proof. We can give either a high level, visual
description of the dynamics of proof simplification, or a low level
description which is suited for the implementation or for fine-grained
analyses of the cut elimination process.

This graphical notation can be formally defined as a system of port
graphs~\cite{AndreiO:PhD,AndreiO:calp}.  The visualisation tools
available for port graphs (see~\cite{AndreiO:PORGY}) can then be used
to study proofs and proof transformations.

Summarising, the main contributions of this paper are:

\begin{itemize}
\item A new graphical notation for proofs in intuitionistic logic
  where one can choose different levels of abstraction depending on
  the kind of analysis that will be done on proofs.
\item By the Curry-Howard isomorphism, the previous point gives a
  graphical notation for the $\lambda$-calculus with the same
  characteristics.
\end{itemize}

\paragraph{Related work} In addition to proof nets and interaction
nets, several graphical representations of proofs have been proposed
in the literature.  We can cite for instance the deduction graphs
defined by Geuvers and Loeb~\cite{GeuversH:natdvg} and Lamping's
sharing graphs~\cite{LampingJ:algolc}.  Deduction graphs are a
generalisation of natural deduction and Fitch style flag deduction;
they have both nodes and boxes. The latter are collections of nodes
that form a node themselves, and in this sense they are related to
Milner's bigraphs~\cite{MilnerR:bigraphs}, where the place graph
describes the nesting of nodes. However, deduction graphs do not have
an explicit way to represent sharing; they are not intended as a
notation for fine-grained analysis of resource management in proof
normalisation.  Sharing graphs (introduced in~\cite{LampingJ:algolc}
and further developed by Asperti and Guerrini~\cite{AspertiA:optifp})
were presented as a solution for the implementation of L\'evy's notion
of optimal reduction in the $\lambda$-calculus, and, as their name
suggests, emphasise the sharing of subexpressions: sharing is
explicit.

The rest of this paper is organised as follows. In the following
section we briefly recall intuitionistic logic, proof nets,
interaction nets and the notions of port graph and port graph
rewriting. In Sections 3 and 4 we give the visual representation of
the logic and the $\lambda$-calculus, respectively. In
section~\ref{sec:disc} we discuss the significance of this work, and
speculate on future work.

\section{Background}
In this section we recall (minimal) intuitionistic logic and sketch
some of the ideas behind proof nets and interaction nets, on which our
work is built upon.  To formalise the notation, we recall the notion
of port graph and port graph rewriting
from~\cite{AndreiO:PhD,AndreiO:calp}.  For more details on linear
logic and proof nets, we refer the reader to \cite{GirardJY:linl}. For
details and examples of interaction nets we refer
to~\cite{LafontY:intn}.

\paragraph{Intuitionistic logic}
In Figure~\ref{fig:ind} we give the natural deduction in sequent form
presentation of the logic. We give explicitly the structural rules,
which helps understanding the graphical notation later.
Adding $\lor$ is straightforward, but not included in this paper.

\begin{figure}[htbp]

\textbf{Identity and Structural Group:}
\[
\begin{prooftree}
 \justifies
 A\vdash A
 \axiom
 \end{prooftree}
\qquad
\begin{prooftree}
\Gamma,A,B,\Delta\vdash C
\justifies
\Gamma,B,A,\Delta\vdash C
\Ex
\end{prooftree}
\qquad
\begin{prooftree}
\Gamma \vdash B
\justifies
\Gamma, A \vdash B
\weakening
\end{prooftree}
\qquad
\begin{prooftree}
\Gamma, A, A \vdash B
\justifies
\Gamma,A \vdash B
\contraction
\end{prooftree}
\]
\textbf{Logical Group:}
\[
\begin{prooftree}
\Gamma\vdash A \qquad \Delta\vdash B
\justifies
\Gamma, \Delta\vdash A \land B
\andintro
\end{prooftree}
\qquad
\begin{prooftree}
\Gamma\vdash A \land B
\justifies
\Gamma\vdash A
\andeliml
\end{prooftree}
\qquad
\begin{prooftree}
\Gamma\vdash A \land B
\justifies
\Gamma\vdash B
\andelimr
\end{prooftree}
\]

\[
\begin{prooftree}
\Gamma,A\vdash B
\justifies
\Gamma\vdash A \limp B
\impintro
\end{prooftree}
\qquad
\begin{prooftree}
\Gamma\vdash A \limp B \qquad \Delta\vdash A
\justifies
\Gamma,\Delta\vdash B
\impelim
\end{prooftree}
\]
\caption{Intuitionistic Natural Deduction}
\label{fig:ind}
\end{figure}

The main computational interest for us is the normalisation procedure,
which transforms proofs by eliminating redundancies (called detours by
Prawitz~\cite{Prawitz:natdeduc}). The main cases are defined below,
 showing how the introduction of a connective followed by the
elimination of that same occurrence of the connective can be
transformed into a proof without the two rules. For simplicity, we
have ignored details of permutations of rules that might need to be
applied so that one of the following rules can be applied.

\begin{definition}(One Step Normalisation)\label{def:norm}
\begin{itemize}
\item $({\land}\I)$ followed by $({\land}\E_1)$: 
below  the double line indicates that (W) may be applied zero or
many times. There is a similar case for $({\land}\I)$ followed by
$({\land}\E_2)$.
\begin{center}
\renewcommand{\tabularxcolumn}[1]{>{\arraybackslash}m{#1}}
\begin{tabular}{lrr}
\begin{prooftree}
\prooftree
\prooftree
\pi_1\justifies\Gamma \vdash  A 
\endprooftree
\qquad
\prooftree
\pi_2\justifies\Delta \vdash B
\endprooftree
\justifies
\Gamma,\Delta \vdash A \land B
\andintro
\endprooftree
\justifies
\Gamma,\Delta \vdash A
\andeliml
\end{prooftree}
&
becomes 
&
\begin{prooftree}
\prooftree
\pi_1
\justifies
\Gamma \vdash A
\endprooftree

\Justifies
\Gamma,\Delta \vdash A
\weakening
\end{prooftree}
\end{tabular}
\end{center}

\item $({\limp}\I)$ followed by $({\limp}\E)$: $\pi_1^\prime$ is the
  proof $\pi_1$ where all axioms $ A \vdash A $ are replaced
  (substituted) by a proof of $\pi_2$. 
\begin{center}
\renewcommand{\tabularxcolumn}[1]{>{\arraybackslash}m{#1}}
\begin{tabular}{lrr}
\begin{prooftree}
\prooftree
\prooftree
\pi_1 \justifies\Gamma, A \vdash  B 
\endprooftree
\justifies
\Gamma \vdash A \limp B
\impintro
\endprooftree
\qquad
\prooftree
\pi_2 \justifies\Delta \vdash A
\endprooftree 
\justifies
\Gamma,\Delta \vdash B
\impelim
\end{prooftree}
&
becomes&
\begin{prooftree}
\pi_1^\prime
\justifies
\Gamma,\Delta \vdash B
\end{prooftree}
\end{tabular}
\end{center}
\end{itemize}
\end{definition}
In this presentation, the notion of substitution is important, and is
quite difficult to understand. An advantage of the visual
representation given below is that it clarifies this notion.

\paragraph{Proof Nets and Interaction Nets}\label{sec:in}
Proof nets were introduced as the graphical syntax for linear logic.
One of the motivations for the study of graphical presentations is to
free us from inessential permutations in proofs, as mentioned in the
Introduction.  Proof nets work very well for the multiplicative
fragment of the logic, but less well for the other fragments.  For
instance, for the exponentials, more complicated machinery is needed,
which takes us away from a uniform visual notation. More precisely,
exponentials are represented using boxes to group parts of the graph,
shown as a dotted line in the diagram below:

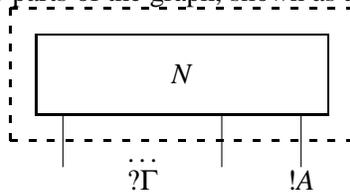
\begin{net}{130}{60}
\putbox{10}{25}{110}{30}{$N$}
\dashline{3}(0,15)(0,65)
\dashline{3}(0,65)(130,65)
\dashline{3}(130,65)(130,15)
\dashline{3}(0,15)(130,15)
\putVline{20}{5}{20}
\putVline{80}{5}{20}
\putVline{110}{5}{20}
\puttext{50}{0}{$?\Gamma$}
\puttext{50}{7}{$\cdots$}
\puttext{110}{0}{$!A$}
\end{net}
Boxes work at a different level to the other nodes in the graph, leading to 
a two-level syntax.
Interaction nets, on the other-hand, encode the box machinery in the 
same notation, as shown below.

An interaction net system~\cite{LafontY:intn} is specified by giving a
set $\Sigma$ of symbols, and a set $\IR$ of interaction rules. Each
symbol $\alpha\in \Sigma$ has an associated (fixed) \emph{arity}. An
occurrence of a symbol $\alpha\in\Sigma$ will be called an
\emph{agent}. If the arity of $\alpha$ is $n$, then the agent has
$n+1$ \emph{ports}: a distinguished one called the \emph{principal
port} depicted by an arrow, and $n$ \emph{auxiliary ports} labelled
$x_1,\ldots,x_n$ corresponding to the arity of the symbol. Such an
agent will be drawn in the following way:
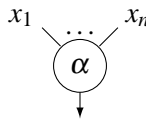
\begin{net}{40}{41}
\putalpha{20}{20}
\putDvector{20}{10}{10}
\putline{12.6}{27.4}{-1}{1}{7}
\putline{27.4}{27.4}{1}{1}{7}
\puttext{20}{32}{$\cdots$}
\put(2.6,35){\makebox(0,0)[br]{$x_1$}}
\put(37.4,35){\makebox(0,0)[bl]{$x_n$}}
\end{net}
A net $N$ built on $\Sigma$ is a graph (not necessarily connected)
with agents at the vertices. The edges of the graph connect agents
together at the ports such that there is only one edge at every
port. 

A pair of agents $(\alpha,\beta)\in \Sigma\times\Sigma$ connected
together on their principal ports is called an \emph{active pair}; the
interaction net analog of a redex. An interaction rule
$((\alpha,\beta) \Lra N) \in \IR$ replaces an occurrence of the active
pair $(\alpha,\beta)$ by a net $N$. The rule must satisfy two
conditions: all free ports are preserved during reduction, 
and there is at most one rule for each pair of agents. 

Boxes are encoded in interaction nets using extra nodes, as shown below:
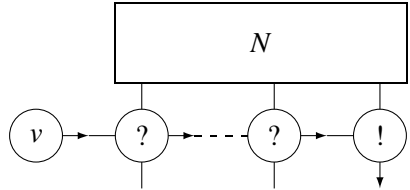
\begin{net}{150}{70}
\putfullprom{0}{0}{$N$}
\end{net}
Due to the constraints in the rules, interaction nets are easy to
implement, but extra rules are needed for management: the details of
copying and erasing for instance must be given in full detail. In
addition, each node in an interaction net has a unique principal port,
and this is fixed for each kind of agent, which means that the
reduction system is fixed as part of the encoding.

Our approach in this paper is to put forward a hybrid notation between
proof nets and interaction nets to get the best from each.  We will 
be able to use  interactive tools developed for interaction nets, such as
PORGY~\cite{AndreiO:PORGY} in order to visualise the encodings of
proofs. In fact, PORGY deals with port graphs, a class of graphs that
is more general than interaction nets, and which can be used to
formalise our hybrid notation. We recall port graphs below.

\paragraph{Port graphs}
A port graph~\cite{AndreiO:PhD} is a graph where nodes have
explicit connection points for edges, called {\em ports}.  Port graphs
were first identified as an abstract view of proteins and molecular
complexes.

Let $\mathscr{N}$ and $\mathscr{P}$ be two disjoint sets of node names
and port names respectively. A {\em p-signature} over $\mathscr{N}$
and $\mathscr{P}$ is a mapping $\nabla : \mathscr{N} \ra
2^{\mathscr{P}}$ which associates a finite
set of port names to a node name.
A p-signature can be extended with variables: $\nabla^\X :
\mathscr{N}\cup \X_{\mathscr{N}} \ra 2^{\mathscr{P}\cup
  \X_{\mathscr{P}}}$, where $\X_{\mathscr{P}}$ and $\X_{\mathscr{N}}$
are two disjoint sets of port name variables and node name variables
respectively.  A {\em labelled port graph} over a {\em p-signature}
$\nabla^\X$ is a tuple $G=\langle V_G,lv_G,E_G, le_G\rangle$ where:
$V_G$ is a finite set of nodes; 
$lv_G:V_G \ra \mathscr{N}\cup \X_{\mathscr{N}}$ is an  injective labelling 
function for nodes; 
$E_G\subseteq \{ \langle (v_1,p_1), (v_2,p_2)\rangle \ | \ v_i\in V_G, p_i\in
\nabla(lv_G(v_i))\cup \X_{\mathscr{P}} \}$ is a finite multiset of
edges; 
$le_G:E_G\ra (\mathscr{P}\cup \X_{\mathscr{P}})\times (\mathscr{P}\cup
\X_{\mathscr{P}})$ is an injective labelling function for edges such
that $le_G(\langle (v_1,p_1), (v_2,p_2)\rangle) = (p_1, p_2)$.  A port
may be associated to a state (for instance, active/inactive or
principal/auxiliary); this is formalised using a mapping from ports to
port states. Similarly, nodes can also have associated properties
(like colour or shape that can be used for visualisation purposes).

Let $G$ and $H$ be two port graphs  over the same p-signature
$\nabla^\X$. A {\em port graph morphism} $f:G\ra H$ maps
elements of $G$ to elements of $H$ preserving sources and targets
of edges, constant node names and associated port name sets, up to 
variable renaming.  We say that $G$ and $H$ are {\em isomorphic} if
$f:V_G\times\nabla(lv_G(V_G)) \ra V_H\times \nabla(lv_H(V_H))$ is
bijective.

A {\em port graph rewrite rule} $L \Ra R$ is itself represented as a
port graph consisting of two port graphs $L$ and $R$ over the same
p-signature and one special node $\Ra$, called {\em arrow node}
connecting them.  $L$ and $R$ are called the {\em left-} and {\em
  right-hand side} respectively.  The arrow node is used to represent
the interface of the rule; it has the following characteristics: for
each port $p$ in $L$, to which corresponds a non-empty set of ports
$\{p_1,\ldots, p_n\}$ in $R$, the arrow node has a unique port $r$ and
the incident directed edges $(p,r)$ and $(r,p_i)$, for all
$i=1,\ldots,n$; all ports from $L$ that are deleted in $R$ are
connected to the {\em black hole} port of the arrow node.  When the
correspondence between ports in the left- and right-hand side of the
rule is obvious we omit the ports and edges involving the arrow node.
In this way, we avoid dangling edges after rewriting (for more details
on graph rewriting we refer the reader
to~\cite{HabelMP01,Corradini:handbook}).

Let $L\Ra R$ be a port graph rewrite rule and $G$ a port graph such
that there is an injective port graph morphism $g$ from $L$ to $G$;
hence $g(L)$ is a subgraph of $G$.  A \emph{rewriting step} on $G$
using $L\Ra R$, written $G\ra_{L\Ra R}G'$, transforms $G$ into a new
graph $G'$ obtained from $G$ by replacing the subgraph $g(L)$ of $G$
by $g(R)$, and connecting $g(R)$ to the rest of the graph as specified
in the arrow node. We call $g(L)$ a redex. If there is no such
injective morphism, we say that $G$ is {\em irreducible} by $L\Ra R$.
Given a finite set $\R$ of rules, a port graph $G$ {\em rewrites} to
$G'$, denoted by $G\Ra_{\R}G'$, if there is a rule $r$ in $\R$ such
that $G\Ra_{r}G'$.  This induces a transitive relation on port graphs,
denoted by $\Ra_{\R}^*$. A port graph on which no rule is applicable
is in \emph{normal form}.

\section{Graphs from Proofs}
In this section we give a graphical representation of proofs in
intuitionistic logic.  The general idea is to interpret a proof $\pi$
of $\Gamma\vdash A$ as a port graph $\Gr{\pi}$ with edges representing
formulas in the following way (the second alternative borrows a
notation from electronic circuits):
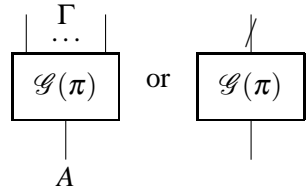
\begin{center}
\begin{picture}(150,70)
\put(30,25){\framebox(40,25){$\Gr{\pi}$}}
\put(50,25){\line(0,-1){15}}
\put(50,0){\makebox(0,0)[b]{$A$}}
\put(35,50){\line(0,1){15}}
\put(65,50){\line(0,1){15}}
\put(50,55){\makebox(0,0){$\cdots$}}
\put(50,65){\makebox(0,0){$\Gamma$}}

\puttext{85}{40}{or}
\putnline{120}{50}{15}
\putbox{100}{25}{40}{25}{$\Gr{\pi}$}
\putVline{120}{10}{15}
\end{picture}
\end{center}

The nodes of the graph represent rules in the logic, edges will be
attached to ports which are optionally labelled with formulas. We will
explicitly distinguish the conclusion port when it is relevant. Node names
will be introduced on demand in the translation below. Later we will
see that we need some additional control nodes, so that there will not
be a 1-1 correspondence between logical rules and node names.  We give a
translation inductively over the structure of the proof, and refer to
Figure~\ref{fig:ind} for the rules that we are translating.

\begin{itemize}
\item If $\pi$ is an Axiom $A \vdash A$, then $\Gr{\pi}$ is simply a node 
$Ax$ with two ports, both with label $A$. in the diagrams we omit this node 
and draw simply a line as it is often done in proof nets.
\begin{net}{10}{30}
\puttext{5}{30}{$A$}
\putVline{5}{5}{20}
\puttext{5}{0}{$A$}
\end{net}
\item Exchange. If $\pi_1$ is a proof ending in $\Gamma, A, B, \Delta
\vdash C$, then we can build a proof $\pi$ of $\Gamma, B, A, \Delta \vdash
C$, using the exchange rule, and a graph $\Gr{\pi}$ where the exchange rule is encoded by exchanging two edges:
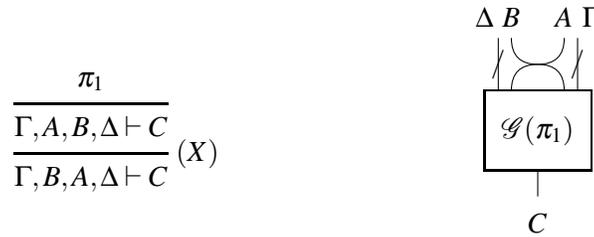
\begin{center}
\newcolumntype{Z}{>{\centering\arraybackslash}X}
\renewcommand{\tabularxcolumn}[1]{>{\arraybackslash}m{#1}}
\begin{tabularx}{.7\textwidth}{ZZ}
\begin{prooftree}
\[
\pi_1
\justifies
\Gamma,A,B,\Delta\vdash C
\]
\justifies
\Gamma,B,A,\Delta\vdash C
\Ex
\end{prooftree}
&
\begin{mnet}{40}{73}
\putbox{0}{20}{40}{30}{$\Gr{\pi_1}$}
\putVline{20}{10}{10}
\puttext{20}{0}{$C$}
\putVline{5}{50}{20}
\putVline{35}{50}{20}
\qbezier(10,50)(10,60)(20,60)
\qbezier(20,60)(30,60)(30,70)
\qbezier(30,50)(30,60)(20,60)
\qbezier(20,60)(10,60)(10,70)
\puttext{0}{78}{$\Delta$}
\puttext{40}{78}{$\Gamma$}
\puttext{10}{78}{$B$}
\puttext{30}{78}{$A$}
\puttext{5}{60}{$/$}
\puttext{35}{60}{$/$}
\end{mnet}
\end{tabularx}
\end{center}

\item Weakening.  If $\pi_1$ is a proof ending in $\Gamma\vdash B$,
then we can build a proof $\pi$ of $\Gamma,A\vdash B$ using the
weakening rule, and a graph $\Gr{\pi}$ as follows, where we explicitly
mark the erasing port in the node $W$:
\begin{center}
\newcolumntype{Z}{>{\centering\arraybackslash}X}
\renewcommand{\tabularxcolumn}[1]{>{\arraybackslash}m{#1}}
\begin{tabularx}{.7\textwidth}{ZZ}
\[
\begin{prooftree}
\[
\pi_1
\justifies
\Gamma \vdash B
\]
\justifies
\Gamma, A \vdash B
\weakening
\end{prooftree}
\]
&
\begin{mnet}{60}{35}
\putbox{0}{20}{30}{20}{$\Gr{\pi_1}$}
\putVline{15}{40}{15}
\puttext{15}{47}{$/$}
\putVline{15}{10}{10}
\puttext{15}{0}{$B$}
\puttext{50}{40}{$\bullet$}
\putagent{50}{30}{$W$}
\putVline{50}{40}{10}
\puttext{50}{55}{$A$}
\puttext{10}{55}{$\Gamma$}
\end{mnet}
\end{tabularx}
\end{center}

\item Contraction. If $\pi_1$ is a proof ending in $\Gamma, A,A\vdash B$,
then we can build a proof $\pi$ of $\Gamma,A\vdash B$ using the
contraction rule, and a graph $\Gr{\pi}$, where we explicitly mark the
copying port in the node $C$:
\begin{center}
\newcolumntype{Z}{>{\centering\arraybackslash}X}
\renewcommand{\tabularxcolumn}[1]{>{\arraybackslash}m{#1}}
\begin{tabularx}{.7\textwidth}{ZZ}
\[
\begin{prooftree}
\[
\pi_1
\justifies
\Gamma, A, A \vdash B
\]
\justifies
\Gamma,A \vdash B
\contraction
\end{prooftree}
\]
&
\begin{net}{40}{60}
\putnline{5}{35}{35}
\putVline{27.3}{35}{7.7}
\putVline{12.7}{35}{7.7}
\putagent{20}{50}{$C$}
\puttext{20}{60}{$\bullet$}
\putVline{20}{60}{10}
\puttext{0}{65}{$\Gamma$}
\puttext{27}{65}{$A$}
\putbox{0}{15}{35}{20}{$\Gr{\pi_1}$}
\putVline{20}{7}{8}
\puttext{20}{0}{$B$}

\end{net}

\end{tabularx}
\end{center}

\item  If $\pi_1$ is a proof ending in $\Gamma \vdash A$ and $\pi_2$ is
a proof ending in $\Delta \vdash B$, then we can build a proof $\pi$
ending with $\Gamma, \Delta\vdash A\land B$ using the ${\land}\I$
rule, and a graph $\Gr{\pi}$, where we introduce a new node $\land\I$
corresponding to the rule, and explicitly mark its conclusion port:
\begin{center}
\newcolumntype{Z}{>{\centering\arraybackslash}X}
\renewcommand{\tabularxcolumn}[1]{>{\arraybackslash}m{#1}}
\begin{tabularx}{.7\textwidth}{ZZ}
\[
\begin{prooftree}
\[
\pi_1
\justifies
\Gamma\vdash A
\]
\qquad
\[
\pi_2
\justifies
\Delta\vdash B
\]
\justifies
\Gamma,\Delta\vdash  A\land B
\using ({\land}\I)
\end{prooftree}
\]
&
\begin{mnet}{150}{90}
\put(30,45){\framebox(30,30){$\Gr{\pi_1}$}}
\put(35,75){\line(0,1){10}}
\put(55,75){\line(0,1){10}}
\put(45,80){\makebox(0,0){$\cdots$}}
\put(45,85){\makebox(0,0)[b]{$\Gamma$}}
\put(80,45){\framebox(30,30){$\Gr{\pi_2}$}}
\put(85,75){\line(0,1){10}}
\put(105,75){\line(0,1){10}}
\put(95,80){\makebox(0,0){$\cdots$}}
\put(95,85){\makebox(0,0)[b]{$\Delta$}}
\putcircle{70}{25}{20}{${\land}\I$}
\put(70,15){\line(0,-1){10}}
\puttext{70}{15}{$\bullet$}
\put(70,0){\makebox(0,0)[t]{$A\land B$}}
\put(45,45){\line(1.3,-1){17}}
\put(95,45){\line(-1.3,-1){17}}
\end{mnet}
\end{tabularx}
\end{center}
\item If $\pi_1$ is a proof ending in  $\Gamma\vdash A\land B$, then we can
build a proof $\pi$ of $\Gamma\vdash A$ using the ${\land}\E_1$ rule, and a graph $\Gr{\pi}$, where we introduce a new node $\land\E_1$:
\begin{center}
\newcolumntype{Z}{>{\centering\arraybackslash}X}
\renewcommand{\tabularxcolumn}[1]{>{\arraybackslash}m{#1}}
\begin{tabularx}{.7\textwidth}{ZZ}
\[
\begin{prooftree}
\[
\pi_1
\justifies
\Gamma\vdash A\land B
\]
\justifies
\Gamma\vdash  A
\using ({\land}\E_1)
\end{prooftree}
\]
&
\begin{mnet}{30}{80}
\putbox{0}{45}{30}{30}{$\Gr{\pi_1}$}
\putagent{15}{25}{$\land\E_1$}
\putVline{15}{7}{8}
\putVline{15}{35}{10}
\puttext{15}{15}{$\bullet$}
\puttext{15}{0}{$A$}
\putVline{5}{75}{10}
\putVline{25}{75}{10}
\puttext{15}{78}{$\cdots$}
\puttext{15}{85}{$\Gamma$}
\end{mnet}
\end{tabularx}
\end{center}

\item If $\pi_1$ is a proof ending in  $\Gamma\vdash A\land B$, then we can
build a proof $\pi$ of $\Gamma\vdash B$ using the ${\land}\E_2$ rule, and a graph $\Gr{\pi}$ where we introduce a new node $\land\E_2$. The graph is similar to the previous case, except that we use a node $\land\E_2$ instead of $\land\E_1$, and conclude $B$ instead of $A$.

\item  If $\pi_1$ is a proof ending in $\Gamma,A\vdash B$ then we can build a proof
$\pi$ of $\Gamma \vdash A\limp B$ using the ${\limp}\I$ rule, and a graph $\Gr{\pi}$:
\begin{center}
\newcolumntype{Z}{>{\centering\arraybackslash}X}
\renewcommand{\tabularxcolumn}[1]{>{\arraybackslash}m{#1}}
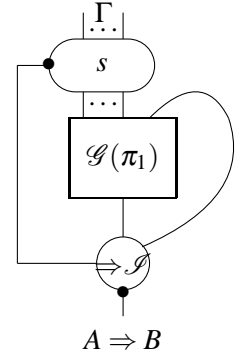
\begin{tabularx}{.7\textwidth}{ZZ}
\begin{prooftree}
\[
\pi_1
\justifies
\Gamma,A\vdash B
\]
\justifies
\Gamma\vdash A \limp B
\impintro
\end{prooftree}
&
\begin{mnet}{200}{120}
\put(70,60){\framebox(40,30){$\Gr{\pi_1}$}}
\put(75,90){\line(0,1){10}}
\put(90,90){\line(0,1){10}}
\put(82.5,95){\makebox(0,0){$\cdots$}}
\put(82.5,130){\makebox(0,0){$\Gamma$}}
\put(82.5,123){\makebox(0,0){$\cdots$}}
\put(90,60){\line(0,-1){15}}
\putcircle{90}{35}{20}{${\limp}\I$}
\put(90,25){\line(0,-1){10}}
\puttext{90}{24}{$\bullet$}
\put(90,10){\makebox(0,0)[t]{$A\limp B$}}
\qbezier(100,90)(120,110)(130,90)
\qbezier(97.3,41.3)(140,60)(130,90)
\put(82,110){\oval(40,20)}
\putHline{50}{35}{30}
\putVline{50}{35}{75}
\putHline{50}{110}{12}
\put(75,120){\line(0,1){10}}
\put(90,120){\line(0,1){10}}
\puttext{62}{110}{$\bullet$}
\puttext{82}{110}{$s$}
\end{mnet}
\end{tabularx}
\end{center}
where we have introduced a new node $\limp\!\!\!\I$ corresponding to the
rule, and a second node $s$ with variable arity (formally, a family of nodes)
that represents the scope of the rule:
when the assumption $A$ is discharged we mark all the other
assumptions. This structure plays a role to visually represent scope,
but more importantly, it will play a crucial role in the dynamics of
cut-elimination that we give later. We remark that if $\Gamma$ is
empty (i.e.\ the proof is closed), then we do not need this extra
node, and will just draw the $\limp\I$ node.
Note also that the node  $s$ does not correspond to a connective.

\item  If $\pi_1$ is a proof ending in $\Gamma\vdash A\limp B$ and
$\pi_2$ is a proof ending in $\Delta\vdash A$, then we can build a
proof $\pi$ of $\Gamma,\Delta \vdash B$ using the $\limp\E$ rule, and a graph $\Gr{\pi}$, where we introduce a new node $\limp\E$:
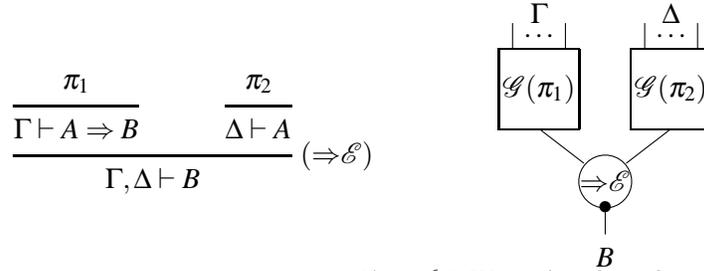
\begin{center}
\newcolumntype{Z}{>{\centering\arraybackslash}X}
\renewcommand{\tabularxcolumn}[1]{>{\arraybackslash}m{#1}}
\begin{tabularx}{.7\textwidth}{ZZ}
\begin{prooftree}
\[
\pi_1
\justifies
\Gamma\vdash A \limp B
\]
\qquad
\[
\pi_2
\justifies
\Delta\vdash A
\]
\justifies
\Gamma,\Delta\vdash  B
\using ({\limp}\E)
\end{prooftree}
&
\begin{mnet}{150}{90}
\put(30,45){\framebox(30,30){$\Gr{\pi_1}$}}
\put(35,75){\line(0,1){10}}
\put(55,75){\line(0,1){10}}
\put(45,80){\makebox(0,0){$\cdots$}}
\put(45,85){\makebox(0,0)[b]{$\Gamma$}}
\put(80,45){\framebox(30,30){$\Gr{\pi_2}$}}
\put(85,75){\line(0,1){10}}
\put(105,75){\line(0,1){10}}
\put(95,80){\makebox(0,0){$\cdots$}}
\put(95,85){\makebox(0,0)[b]{$\Delta$}}
\putcircle{70}{25}{20}{${\limp}\E$}
\put(70,15){\line(0,-1){10}}
\puttext{70}{15}{$\bullet$}
\put(70,0){\makebox(0,0)[t]{$B$}}
\put(45,45){\line(1.3,-1){17}}
\put(95,45){\line(-1.3,-1){17}}
\end{mnet}
\end{tabularx}
\end{center}
\end{itemize}
The set of node names introduced for the logic is $\NLog= \{C,W,\land\I,\land\E_1,\land\E_2,\limp\!\!\!\I, \limp\!\!\!\E, s\}$.

\subsection{Example}
To illustrate the translation, we give here an example proof and the
corresponding graph. The proof is of $A \land B \vdash B\land A$,
which shows that $\land$ is commutative, and in the graph the formulae
on the edges can be read from the proof:
\begin{center}
\newcolumntype{Z}{>{\centering\arraybackslash}X}
\renewcommand{\tabularxcolumn}[1]{>{\arraybackslash}m{#1}}
\begin{tabularx}{\textwidth}{ZZZ}
\begin{prooftree}
\[\[\[
\justifies
A\land B \vdash A\land B \axiom\]
\justifies
A\land B\vdash B
\andelimr
\]
\qquad
\[\[\justifies A\land B \vdash A\land B\axiom\]
\justifies
A\land B\vdash A
\andeliml
\]
\justifies
A\land B,A \land B\vdash B\land A
\andintro
\]
\justifies
A \land B \vdash B\land A
\contraction
\end{prooftree}
&
\begin{mnet}{80}{85}
\putagent{40}{20}{$\land\I$}
\putagent{40}{80}{$C$}
\putagent{10}{50}{$\land\E_2$}
\putagent{70}{50}{$\land\E_1$}
\putVline{40}{0}{10}
\putVline{40}{90}{10}
\putline{47.3}{27.3}{1}{1}{15.9}
\putline{17.3}{57.3}{1}{1}{15.9}
\putline{17.3}{42.7}{1}{-1}{15.9}
\putline{47.3}{72.7}{1}{-1}{15.9}
\puttext{40}{90}{$\bullet$}
\puttext{40}{10}{$\bullet$}
\puttext{17.3}{42.7}{$\bullet$}
\puttext{63.7}{42.7}{$\bullet$}
\end{mnet}
\end{tabularx}
\end{center}
In this example, we can see the true structure of the underlying
proof. Contraction was the last rule used in the derivation and is the last
rule in the diagram, at the top since the structural rules work on 
the left of the $\vdash$ symbol and at the top of the diagrams (see the
translation above).

Note that we have chosen to use the multiplicative presentation of
intuitionistic logic in the sequent calculus, additive is an
alternative, and this would lead to a different graphical
representation. An essential issue is that the graphical representation is
faithful to this choice. In this example, we invite the reader to
apply one more rule (${\limp}\I$) to give a proof of $\vdash A \land B
\limp B\land A$.

A second example illustrates one of the axioms of intuitionistic
logic: $A\limp B\limp A$. The proof and the graph are the following, where we demonstrate the two different translations of ${\limp}\I$, one closed using just an ${\limp}\I$ node, and one with free variables using an $s$ node.
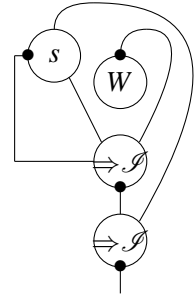
\begin{center}
\newcolumntype{Z}{>{\centering\arraybackslash}X}
\renewcommand{\tabularxcolumn}[1]{>{\arraybackslash}m{#1}}
\begin{tabularx}{\textwidth}{ZZZ}
\[
\begin{prooftree}
\[
\[
\[
\justifies
A\vdash A
\axiom
\]
\justifies
A,B \vdash A
\using (W)
\]
\justifies
A \vdash B\limp A
\using (\limp\I)
\]
\justifies
\vdash A\limp B\limp A
\using (\limp\I)
\end{prooftree}
\]
&
\begin{mnet}{60}{95}
\putagent{40}{20}{${\limp}\I$}
\putagent{40}{50}{${\limp}\I$}
\putagent{40}{80}{$W$}
\putagent{15}{90}{$s$}
\putVline{40}{0}{10}
\putVline{40}{30}{10}
\putVline{0}{50}{40}
\putHline{0}{50}{30}
\putHline{0}{90}{5}
\qbezier(15,100)(15,110)(30,110)
\qbezier(47.2,27.3)(95,110)(30,110)
\qbezier(40,90)(40,100)(50,100)
\qbezier(47.2,57.3)(70,100)(50,100)
\putline{32.85}{57.3}{-1}{2}{12.3}
\puttext{40}{10}{$\bullet$}\puttext{40}{40}{$\bullet$}
\puttext{40}{90}{$\bullet$}
\puttext{5}{90}{$\bullet$}

\end{mnet}
\end{tabularx}
\end{center}

\subsection{Normalisation}\label{sec:graphnorm}

Next we turn to the normalisation process in this graphical setting.
Our aim is to show graph transformations, formalised as port graph
rewrite rules, for each of the normalisation steps given in
Definition~\ref{def:norm}.

\begin{itemize}
\item $({\land}\I)$ followed by $({\land}\E_1)$: we attach the weakening node $W$ to the graph representing
$\pi_2$.  There is a symmetric case for $({\land}\I)$ followed by $({\land}\E_2)$.
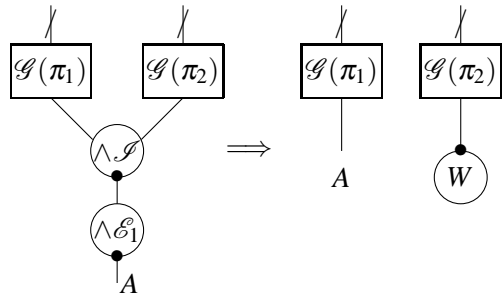
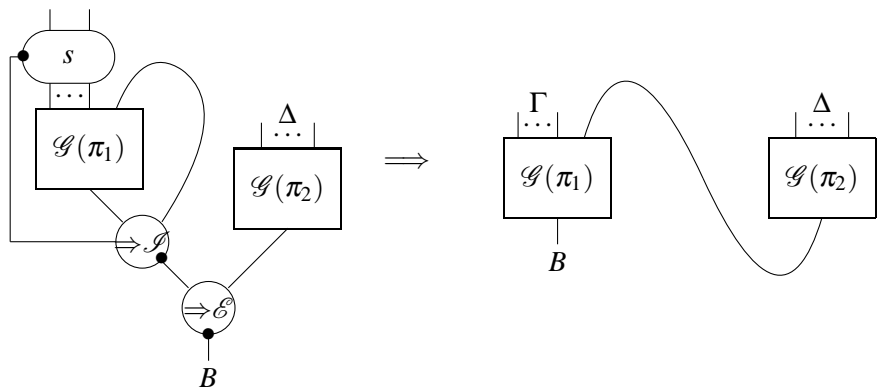
\begin{net}{180}{97}
\putbox{0}{70}{30}{20}{$\Gr{\pi_1}$}
\putbox{50}{70}{30}{20}{$\Gr{\pi_2}$}

\putbox{110}{70}{30}{20}{$\Gr{\pi_1}$}

\putagent{40}{50}{$\land\I$}
\putagent{40}{20}{$\land\E_1$}

\putbox{155}{70}{30}{20}{$\Gr{\pi_2}$}
\putagent{170}{40}{$W$}
\putVline{170}{50}{20}

\putVline{40}{0}{10}
\putVline{40}{30}{10}

\putnline{15}{90}{15}
\putnline{65}{90}{15}
\putnline{125}{90}{15}
\putnline{170}{90}{15}

\putline{15}{70}{1}{-1}{15.8}
\putline{65}{70}{-1}{-1}{15.8}

\putVline{125}{50}{20}
\puttext{125}{40}{$A$}

\puttext{45}{0}{$A$}
\puttext{90}{50}{$\Lra$}
\puttext{170}{50}{$\bullet$}
\puttext{40}{40}{$\bullet$}
\puttext{40}{10}{$\bullet$}

\end{net}
\item $({\limp}\I)$ followed by $({\limp}\E)$:
\begin{center}
\newcolumntype{Z}{>{\centering\arraybackslash}X}
\renewcommand{\tabularxcolumn}[1]{>{\arraybackslash}m{#1}}
\begin{tabularx}{0.8\textwidth}{ZZZ}
\begin{mnet}{150}{143}
\puttext{22}{130}{$s$}
\putHline{0}{130}{5}
\put(22,130){\oval(35,20)}
\put(15,140){\line(0,1){8}}
\put(30,140){\line(0,1){8}}
\puttext{5}{130}{$\bullet$}
\put(10,80){\framebox(40,30){$\Gr{\pi_1}$}}
\put(15,110){\line(0,1){10}}
\put(30,110){\line(0,1){10}}
\put(22.5,115){\makebox(0,0){$\cdots$}}
\puttext{57.3}{53.7}{$\bullet$}
\put(30,80){\line(1,-1){12.6}}
\putcircle{50}{60}{20}{${\limp}\I$}
\qbezier(40,110)(50,135)(65,125)
\qbezier(65,125)(85,110)(57.4,67.5)
\putHline{0}{60}{40}
\putVline{0}{60}{70}
\putcircle{75}{35}{20}{${\limp}\E$}
\put(57.5,52.5){\line(1,-1){10}}
\puttext{75}{25}{$\bullet$}
\put(75,25){\line(0,-1){10}}
\put(75,5){\makebox(0,0)[b]{$B$}}

\put(82.5,42.5){\line(1,1){22.6}}
\put(85,65){\framebox(40,30){$\Gr{\pi_2}$}}
\put(95,95){\line(0,1){10}}
\put(115,95){\line(0,1){10}}
\put(105,100){\makebox(0,0){$\cdots$}}
\put(105,105){\makebox(0,0)[b]{$\Delta$}}
\puttext{150}{90}{$\Lra$}

\end{mnet}
&
\begin{mnet}{140}{85}
\put(0,40){\framebox(40,30){$\Gr{\pi_1}$}}
\put(20,40){\line(0,-1){10}}
\put(20,20){\makebox(0,0)[b]{$B$}}
\put(5,70){\line(0,1){10}}
\put(20,70){\line(0,1){10}}
\put(12.5,75){\makebox(0,0){$\cdots$}}
\put(12.5,80){\makebox(0,0)[b]{$\Gamma$}}
\put(100,40){\framebox(40,30){$\Gr{\pi_2}$}}
\put(110,70){\line(0,1){10}}
\put(130,70){\line(0,1){10}}
\put(120,75){\makebox(0,0){$\cdots$}}
\put(120,80){\makebox(0,0)[b]{$\Delta$}}
\qbezier(30,70)(45,120)(75,55)
\qbezier(75,55)(105,-10)(120,40)
\end{mnet}
\end{tabularx}
\end{center}
This rule performs a substitution.  One of the most beautiful aspects
of the notation is the clear explication of the substitution process,
which can be seen as replacing an axiom in the proof $\pi_1$ with the
proof $\pi_2$.
\end{itemize}

We will also consider the following simplification rule, optimising
proofs when $C$ is followed by $W$ on the same formula. Thus, the
following proof:
\[
\begin{prooftree}
\[
\Gamma,A\vdash B
\justifies
\Gamma,A,A\vdash B
\using (W)
\]
\justifies
\Gamma,A\vdash B
\using (C)
\end{prooftree}
\]
will be represented and simplified by the following:
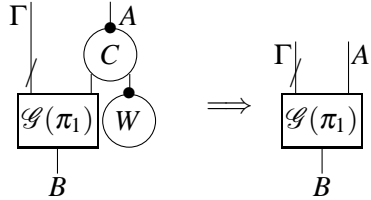
\begin{net}{140}{75}
\putbox{0}{20}{30}{20}{$\Gr{\pi_1}$}
\putVline{15}{10}{10}
\puttext{15}{5}{$B$}
\putnline{5}{40}{35}
\putVline{42.3}{40}{7.7}
\putVline{27.7}{40}{7.7}
\putagent{35}{55}{$C$}
\puttext{35}{65}{$\bullet$}
\putVline{35}{65}{10}
\puttext{0}{70}{$\Gamma$}
\puttext{42}{70}{$A$}
\putw{42}{30}
\puttext{80}{35}{$\Lra$}
\putbox{100}{20}{30}{20}{$\Gr{\pi_1}$}
\putVline{115}{10}{10}
\puttext{115}{5}{$B$}
\putnline{105}{40}{20}
\putVline{125}{40}{20}
\puttext{100}{55}{$\Gamma$}
\puttext{130}{55}{$A$}
\puttext{42}{40}{$\bullet$}

\end{net}

\subsection{Copy and Erase}

The copy and erase nodes, which correspond to the contraction and
weakening rules, require specific rules in order to copy or erase
proofs. This is one of the points where interaction net and proof net
representations differ.  In interaction nets, the copy and erasing
processes are performed by traversing the net after it has been
normalised (in this way, redexes cannot be copied, which ensures a
more efficient implementation). On the other hand, the low level
details of the copy and erase rules obscure the understanding of the
logic.  In the syntax described above, we are free to give global
rules for copy and erase as in proof nets (but the price to pay for
this is a more involved notion of graph rewriting, with an expensive
matching algorithm).  We are also free to choose low level,
interaction net style rules, which have a simple matching algorithm
and are better if we need to analyse the cost of the normalisation
process.

\paragraph{Erasing}
Here we show that the nets arising from the translation
function can be erased, either by global steps on $W$ nodes or using the $\epsilon$ agent to erase locally, thus showing that the weakening cut elimination step above can be fully simulated.  
\begin{itemize}
\item Erasing a node $\alpha \in \NLog$, $\alpha \neq\; \limp\!\!\!\I$:
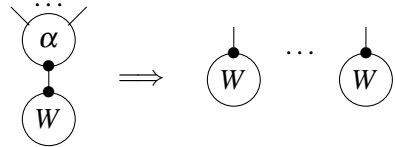
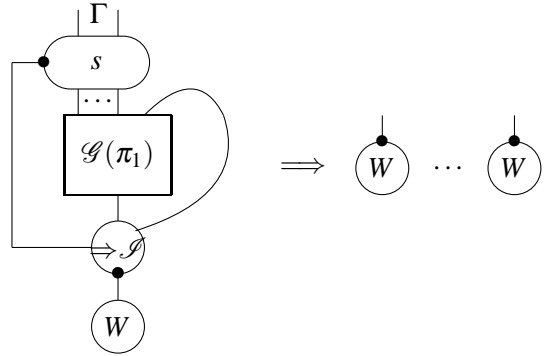
\begin{net}{150}{53}
\putalpha{20}{40}
\putline{12.6}{45.6}{-1}{1}{7}
\putline{27.4}{45.6}{1}{1}{7}
\puttext{20}{53}{$\cdots$}
\putw{20}{10}
\putVline{20}{20}{10}
\puttext{20}{30}{$\bullet$}
\puttext{20}{20}{$\bullet$}

\puttext{55}{25}{$\Lra$}
\putw{90}{25}
\putVline{90}{35}{10}
\puttext{115}{35}{$\cdots$}
\putw{140}{25}
\putVline{140}{35}{10}
\puttext{90}{35}{$\bullet$}
\puttext{140}{35}{$\bullet$}
\end{net}

\item Erasing $\limp\!\!\!\I$:

\begin{net}{240}{115}
\put(70,50){\framebox(40,30){$\Gr{\pi_1}$}}
\puttext{62}{100}{$\bullet$}
\puttext{82}{100}{$s$}
\put(75,80){\line(0,1){10}}
\put(90,80){\line(0,1){10}}
\put(82.5,85){\makebox(0,0){$\cdots$}}
\put(82.5,115){\makebox(0,0)[b]{$\Gamma$}}
\put(90,50){\line(0,-1){10}}
\putcircle{90}{30}{20}{${\limp}\I$}
\put(90,20){\line(0,-1){10}}
\puttext{90}{20}{$\bullet$}
\putw{90}{0}
\qbezier(100,80)(120,100)(130,80)
\qbezier(96.7,36.3)(140,50)(130,80)
\put(82,100){\oval(40,20)}
\putHline{50}{30}{30}
\putVline{50}{30}{70}
\putHline{50}{100}{12}
\put(75,110){\line(0,1){10}}
\put(90,110){\line(0,1){10}}
\puttext{90}{20}{$\bullet$}
\puttext{160}{60}{$\Lra$}
\putw{190}{60}
\putVline{190}{70}{10}
\puttext{215}{60}{$\cdots$}
\putw{240}{60}
\putVline{240}{70}{10}
\puttext{240}{70}{$\bullet$}
\puttext{190}{70}{$\bullet$}
\end{net}
\end{itemize}
This transformation is a global reduction step: it requires that we
identify the graph for $\pi_1$, which in turn relies on a notion of
pattern matching that is not easy to implement (cf.\
subgraph isomorphism~\cite{Ullmann76}).

Alternatively, one can use $\epsilon$ nodes that perform small-step
erasing, in which case the pattern matching algorithm is trivial. In
this case, the previous rule has a left-hand side consisting of just
$W$, $\limp\!\!\!\I$ and $s$, and we have a reduction to:
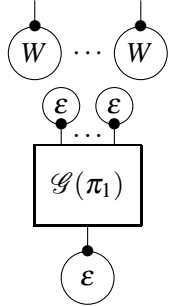
\begin{net}{50}{108}
\puttext{40}{100}{$\bullet$}
\puttext{0}{100}{$\bullet$}
\putw{0}{90}
\putVline{0}{100}{10}
\puttext{20}{90}{$\cdots$}
\putw{40}{90}
\putVline{40}{100}{10}
\putVline{10}{55}{10}
\puttext{20}{59}{$\ldots$}
\puttext{10}{63}{$\bullet$}
\puttext{30}{63}{$\bullet$}
\putmepsilon{10}{70}
\putmepsilon{30}{70}
\putVline{30}{55}{10}
\put(0,25){\framebox(40,30){$\Gr{\pi_1}$}}
\putVline{20}{15}{10}
\putepsilon{20}{5}
\puttext{20}{15}{$\bullet$}
\end{net}
In this way, we can provide a low level definition of weakening which
is better adapted for fine grained analysis of the erasing process.
The rules for $\epsilon$, with $\alpha \in \NLog\cup \{\epsilon\}$,
are below (note that 
if $\alpha$ is $\epsilon$ the right hand side is an empty graph):
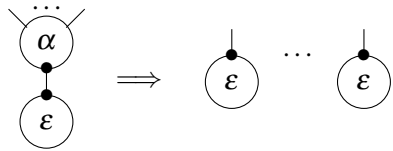
\begin{net}{150}{54}
\putalpha{20}{40}
\putline{12.6}{45.6}{-1}{1}{7}
\putline{27.4}{45.6}{1}{1}{7}
\puttext{20}{53}{$\cdots$}
\putepsilon{20}{10}
\putVline{20}{20}{10}
\puttext{20}{30}{$\bullet$}
\puttext{20}{20}{$\bullet$}
\puttext{55}{25}{$\Lra$}
\putepsilon{90}{25}
\putVline{90}{35}{10}
\puttext{115}{35}{$\cdots$}
\putepsilon{140}{25}
\putVline{140}{35}{10}
\puttext{90}{35}{$\bullet$}
\puttext{140}{35}{$\bullet$}
\end{net}
In this case the cost of erasing the graph (i.e., 
the number of rewriting steps involved) depends on its size.

\begin{lemma}[Erasing]\label{lem:erasing}
Let $N = \Gr{\pi}$, for any proof $\pi$ of $A_1,\ldots,A_{n-1} \vdash
A_n$.  Then using $n$ $\epsilon$ nodes there is a sequence of
rewriting steps that erase $N$, as shown in the following diagram, where
the right-hand side is the empty graph: 

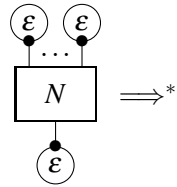
\begin{net}{50}{58}
\putbox{0}{20}{30}{20}{$N$}
\putmepsilon{15}{3}
\putmepsilon{5}{57}
\putmepsilon{25}{57}
\putVline{15}{10}{10}
\putVline{5}{40}{10}
\putVline{25}{40}{10}
\puttext{15}{10}{$\bullet$}
\puttext{5}{50}{$\bullet$}
\puttext{25}{50}{$\bullet$}
\puttext{50}{30}{$\Lra^*$}
\puttext{15}{45}{$\ldots$}
\end{net}
\end{lemma}

\begin{proof} By induction on $\pi$. 
The proof follows a similar structure to the Duplication Lemma
(see Lemma~\ref{lem:duplication} below) that we shall give later. (The case for
duplication is slightly more interesting.)\qed
\end{proof}

\paragraph{Duplication}

Next we address the issue of duplication: specifically, we
show that the graphs arising from the translation function can be
copied, either by global steps on $C$ nodes or by using the $\delta$
agent to copy step-by-step. We first show the rules for the $C$ node.
\begin{itemize}
\item Copying a node $\alpha \neq \; \limp\!\!\!\I$, $\alpha \in \NLog$:
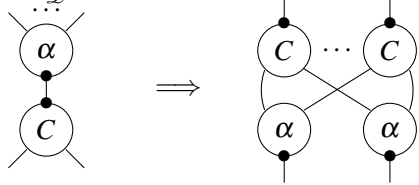
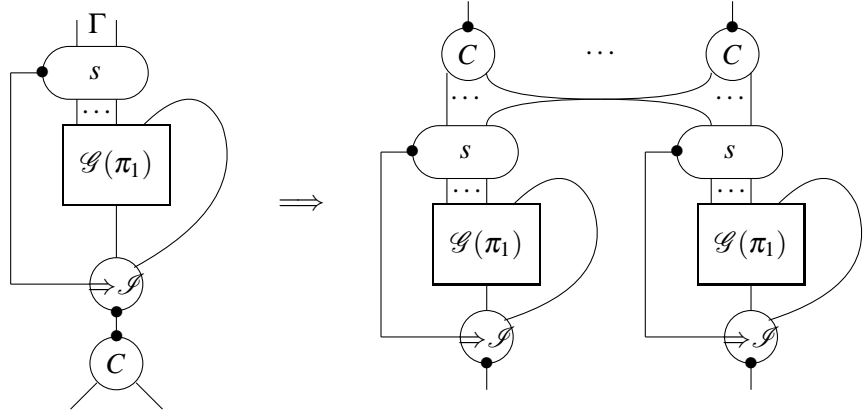
\begin{net}{170}{60}
\puttext{20}{65}{$\cdots$}
\putalpha{20}{50}
\putline{12.6}{57.4}{-1}{1}{7}
\putline{27.4}{57.4}{1}{1}{7}
\puttext{20}{40}{$\bullet$}
\putVline{20}{30}{10}
\puttext{20}{30}{$\bullet$}
\putC{20}{20}
\putline{12.6}{12.6}{-1}{-1}{7}
\putline{27.4}{12.6}{1}{-1}{7}
\puttext{70}{35}{$\Lra$}
\putC{110}{50}
\puttext{110}{60}{$\bullet$}
\putC{150}{50}
\puttext{150}{60}{$\bullet$}
\putVline{150}{60}{10}
\puttext{130}{50}{$\cdots$}
\putVline{110}{60}{10}
\putline{117.4}{27.4}{1.6}{1}{25.5}
\putline{142.6}{27.4}{-1.6}{1}{25.5}
\qbezier(102.6,27.4)(100,35)(102.6,42.6)
\qbezier(157.4,27.4)(160,35)(157.4,42.6)
\puttext{110}{10}{$\bullet$}
\puttext{150}{10}{$\bullet$}
\putalpha{110}{20}
\putalpha{150}{20}
\putVline{110}{0}{10}
\putVline{150}{0}{10}
\end{net}

\item Copying $\limp\!\!\!\I$:
\begin{net}{340}{140}
\put(20,70){\framebox(40,30){$\Gr{\pi_1}$}}
\puttext{32}{120}{$s$}
\put(25,100){\line(0,1){10}}
\put(40,100){\line(0,1){10}}
\put(32.5,105){\makebox(0,0){$\cdots$}}
\put(32.5,135){\makebox(0,0)[b]{$\Gamma$}}
\put(40,70){\line(0,-1){20}}
\putcircle{40}{40}{20}{${\limp}\I$}
\put(40,30){\line(0,-1){10}}
\puttext{40}{29}{$\bullet$}
\putC{40}{10}
\putline{32.6}{2.6}{-1}{-1}{10}
\putline{47.4}{2.6}{1}{-1}{10}
\puttext{40}{20}{$\bullet$}
\qbezier(50,100)(70,120)(80,100)
\qbezier(46.7,46.3)(90,70)(80,100)
\put(32,120){\oval(40,20)}
\putHline{0}{40}{30}
\putVline{0}{40}{80}
\putHline{0}{120}{12}
\put(25,130){\line(0,1){10}}
\put(40,130){\line(0,1){10}}
\puttext{110}{70}{$\Lra$}
\putVline{173}{137}{10}
\puttext{173}{137}{$\bullet$}
\putC{173}{127}
\put(165,100){\line(0,1){20}}
\puttext{172}{110}{$\cdots$}
\qbezier(180,120)(180,110)(223,110)
\qbezier(223,110)(265,110)(265,100)
\putHline{140}{90}{12}
\put(160,40){\framebox(40,30){$\Gr{\pi_1}$}}
\puttext{172}{90}{$s$}
\put(165,70){\line(0,1){10}}
\put(180,70){\line(0,1){10}}
\put(172.5,75){\makebox(0,0){$\cdots$}}
\put(180,40){\line(0,-1){10}}
\qbezier(190,70)(210,90)(220,70)
\qbezier(186.7,26.3)(230,40)(220,70)
\put(172,90){\oval(40,20)}
\putHline{140}{20}{30}
\putVline{140}{20}{70}
\putcircle{180}{20}{20}{${\limp}\I$}
\put(180,10){\line(0,-1){10}}
\puttext{223}{127}{$\cdots$}
\putVline{273}{137}{10}
\puttext{273}{137}{$\bullet$}
\putC{273}{127}
\put(280,100){\line(0,1){20}}
\puttext{0272}{110}{$\cdots$}
\qbezier(180,100)(180,110)(223,110)
\qbezier(223,110)(265,110)(265,120)
\putHline{240}{90}{12}
\put(260,40){\framebox(40,30){$\Gr{\pi_1}$}}
\puttext{272}{90}{$s$}
\put(265,70){\line(0,1){10}}
\put(280,70){\line(0,1){10}}
\put(272.5,75){\makebox(0,0){$\cdots$}}
\put(280,40){\line(0,-1){10}}
\qbezier(290,70)(310,90)(320,70)
\qbezier(286.7,26.3)(330,40)(320,70)
\put(272,90){\oval(40,20)}
\putHline{240}{20}{30}
\putVline{240}{20}{70}
\putcircle{280}{20}{20}{${\limp}\I$}
\put(280,10){\line(0,-1){10}}
\puttext{280}{10}{$\bullet$}
\puttext{180}{10}{$\bullet$}
\puttext{252}{90}{$\bullet$}
\puttext{152}{90}{$\bullet$}
\puttext{12}{120}{$\bullet$}
\end{net}
\end{itemize}
Alternatively, one can use $\delta$ nodes to perform small-step
copying. In this case the left-hand side of the rule is just $C$ with 
$\limp\!\!\!\I$ and $s$, and we have a reduction to:
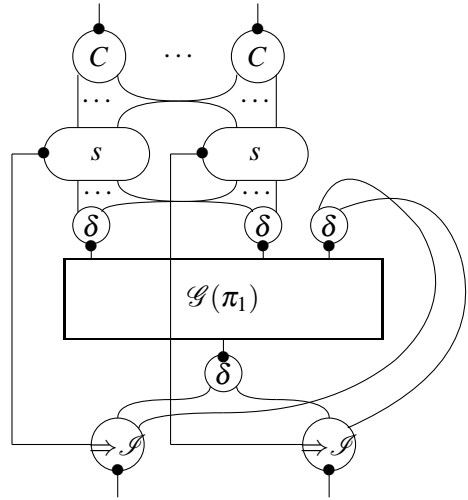
\begin{net}{140}{184}
\putVline{33}{177}{10}
\puttext{33}{177}{$\bullet$}
\putC{33}{167}
\put(25,140){\line(0,1){20}}
\puttext{32}{150}{$\cdots$}
\qbezier(40,160)(40,150)(63,150)
\qbezier(63,150)(85,150)(85,140)
\putHline{0}{130}{12}
\put(20,60){\framebox(120,30){$\Gr{\pi_1}$}}
\puttext{32}{130}{$s$}
\put(25,108){\line(0,1){12}}
\qbezier(40,120)(40,112)(65,112)
\qbezier(34,108)(34,112)(65,112)
\putmdelta{30}{103}
\puttext{30}{95}{$\bullet$}
\putVline{30}{90}{5}
\put(32.5,115){\makebox(0,0){$\cdots$}}
\put(32,130){\oval(40,20)}
\putHline{0}{20}{30}
\putVline{0}{20}{110}
\putcircle{40}{20}{20}{${\limp}\I$}
\put(40,10){\line(0,-1){10}}

\puttext{63}{167}{$\cdots$}

\putVline{93}{177}{10}
\puttext{93}{177}{$\bullet$}
\putC{93}{167}
\put(100,140){\line(0,1){20}}
\puttext{92}{150}{$\cdots$}
\qbezier(40,140)(40,150)(63,150)
\qbezier(63,150)(85,150)(85,160)
\putHline{60}{130}{12}
\puttext{92}{130}{$s$}
\putmdelta{95}{103}
\puttext{95}{95}{$\bullet$}
\putVline{95}{90}{5}
\qbezier(85,120)(85,112)(65,112)
\qbezier(91,108)(91,112)(65,112)
\put(100,108){\line(0,1){12}}
\put(92.5,115){\makebox(0,0){$\cdots$}}
\put(80,60){\line(0,-1){5}}
\puttext{80}{53}{$\bullet$}
\putmdelta{80}{47}
\qbezier(40,30)(40,37)(60,37)
\qbezier(60,37)(75,37)(75,42)
\qbezier(85,42)(85,37)(100,37)
\qbezier(100,37)(120,37)(120,30)

\qbezier(124,109)(160,123)(170,90)
\qbezier(127,26.3)(180,50)(170,90)

\qbezier(116,109)(125,135)(154,102)
\qbezier(154,102)(170,75)(145,52)

\qbezier(48,26)(48,32)(80,32)
\qbezier(80,32)(120,32)(145,52)

\putmdelta{120}{103}
\puttext{120}{95}{$\bullet$}
\putVline{120}{90}{5}
\put(92,130){\oval(40,20)}
\putHline{60}{20}{50}
\putVline{60}{20}{110}
\putcircle{120}{20}{20}{${\limp}\I$}
\put(120,10){\line(0,-1){10}}
\puttext{40}{10}{$\bullet$}
\puttext{120}{10}{$\bullet$}
\puttext{12}{130}{$\bullet$}
\puttext{72}{130}{$\bullet$}

\end{net}
We now introduce the rewrite rules for the $\delta$ nodes.  Let
$\alpha$ be any node in $\NLog$. Then $\delta$ copies the node
and propagates itself to copy the rest of the graph,
 as shown below. The rule for $\delta$ with $\delta$
ends the duplication process. 
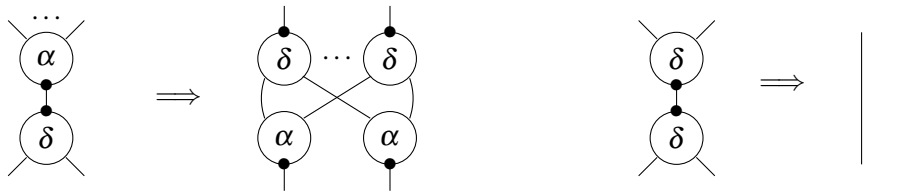
\begin{center}
\begin{mnet}{170}{60}
\puttext{20}{65}{$\cdots$}
\putalpha{20}{50}
\putline{12.6}{57.4}{-1}{1}{7}
\putline{27.4}{57.4}{1}{1}{7}
\puttext{20}{40}{$\bullet$}
\putVline{20}{30}{10}
\puttext{20}{30}{$\bullet$}
\putdelta{20}{20}
\putline{12.6}{12.6}{-1}{-1}{7}
\putline{27.4}{12.6}{1}{-1}{7}
\puttext{70}{35}{$\Lra$}
\putdelta{110}{50}
\puttext{110}{60}{$\bullet$}
\putdelta{150}{50}
\puttext{150}{60}{$\bullet$}
\putVline{150}{60}{10}
\puttext{130}{50}{$\cdots$}
\putVline{110}{60}{10}
\putline{117.4}{27.4}{1.6}{1}{25.5}
\putline{142.6}{27.4}{-1.6}{1}{25.5}
\qbezier(102.6,27.4)(100,35)(102.6,42.6)
\qbezier(157.4,27.4)(160,35)(157.4,42.6)
\puttext{110}{10}{$\bullet$}
\puttext{150}{10}{$\bullet$}
\putalpha{110}{20}
\putalpha{150}{20}
\putVline{110}{0}{10}
\putVline{150}{0}{10}
\end{mnet}
\qquad
\qquad
\qquad
\begin{mnet}{110}{65}
\putVline{20}{30}{10}
\putdelta{20}{50}
\puttext{20}{40}{$\bullet$}
\putline{12.6}{57.4}{-1}{1}{7}
\putline{27.4}{57.4}{1}{1}{7}
\putdelta{20}{20}
\puttext{20}{30}{$\bullet$}
\putline{27.4}{12.6}{1}{-1}{7}
\putline{12.6}{12.6}{-1}{-1}{7}
\puttext{60}{40}{$\Lra$}
\putVline{90}{10}{50}
\putVline{110}{10}{50}
\end{mnet}
\end{center}

\begin{lemma}[Duplication]\label{lem:duplication}
Let $N = \Gr{\pi}$, for any proof $\pi$ of $A_1,\ldots,A_{n-1} \vdash A_n$.  Then using $n$ $\delta$ nodes there is a sequence of rewriting steps that
  duplicates $N$, as shown in the following diagram:
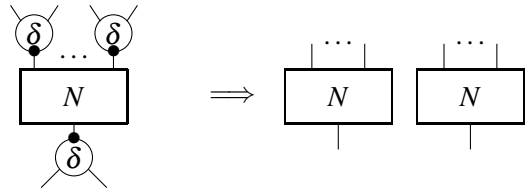
\begin{net}{180}{65}

  \putmdelta{5}{54}
  \puttext{5}{47}{$\bullet$}
  \putline{0}{58.6}{-1}{1.5}{4}
  \putline{10}{58.6}{1}{1.5}{4}
  \put(5,40){\line(0,1){7}}

  \putmdelta{35}{54}
  \puttext{35}{47}{$\bullet$}
  \putline{30}{58.6}{-1}{1.5}{4}
  \putline{40}{58.6}{1}{1.5}{4}
  \put(35,40){\line(0,1){7}}
  \puttext{20}{45}{$\cdots$}

  \putbox{0}{20}{40}{20}{$N$}
  \putVline{20}{15}{5}
  \putmdelta{20}{7}
  \puttext{20}{14}{$\bullet$}
  \putline{14.6}{2.6}{-1}{-1}{7}
  \putline{25.4}{2.6}{1}{-1}{7}

  \puttext{80}{30}{$\Lra$}

  \putbox{100}{20}{40}{20}{$N$}
  \putVline{110}{40}{10}
  \putVline{130}{40}{10}
  \puttext{120}{50}{$\cdots$}
  \putVline{120}{10}{10}

  \putbox{150}{20}{40}{20}{$N$}
  \putVline{160}{40}{10}
  \putVline{180}{40}{10}
  \puttext{170}{50}{$\cdots$}
  \putVline{170}{10}{10}
\end{net}
\end{lemma}

\begin{proof}
By induction on the depth of the proof $\pi$. We show the cases for Axiom and $({\land}\I)$.
\begin{itemize}
\item Axiom:
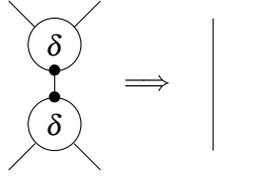
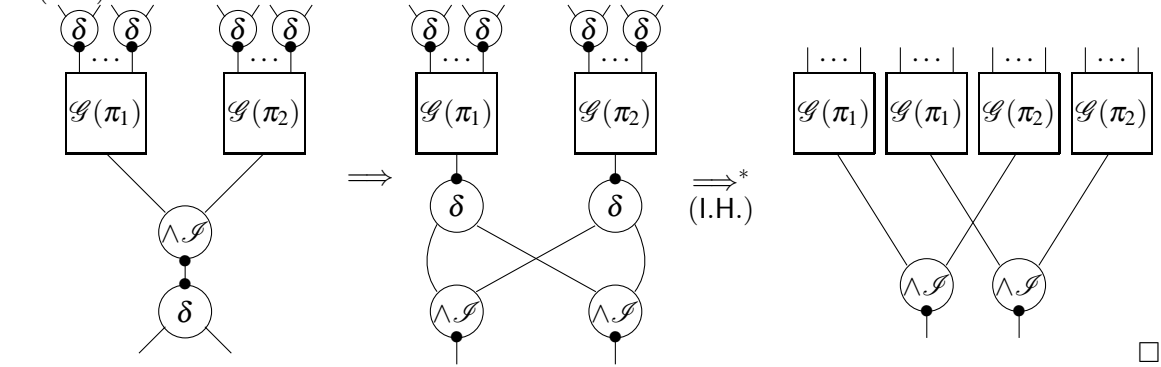
\begin{net}{90}{50}
\putline{2.6}{46.6}{-1}{1}{10}
\putline{17.4}{46.6}{1}{1}{10}
\putdelta{10}{40}
\puttext{10}{30}{$\bullet$}
\putVline{10}{20}{10}
\puttext{10}{20}{$\bullet$}
\putdelta{10}{10}
\putline{2.6}{2.6}{-1}{-1}{10}
\putline{17.4}{2.6}{1}{-1}{10}
\puttext{45}{25}{$\Lra$}
\putVline{70}{0}{50}
\putVline{90}{0}{50}
\end{net}

\item $({\land}\I):$
\begin{center}
\begin{mnet}{130}{110}
\put(0,60){\framebox(30,30){$\Gr{\pi_1}$}}
\putmdelta{5}{107}
\puttext{5}{100}{$\bullet$}
\putline{0}{111.6}{-1}{1.5}{3}
\putline{10}{111.6}{1}{1.5}{3}
\put(5,90){\line(0,1){10}}
\putmdelta{25}{107}
\puttext{25}{100}{$\bullet$}
\putline{20}{111.6}{-1}{1.5}{3}
\putline{30}{111.6}{1}{1.5}{3}
\put(25,90){\line(0,1){10}}
\put(15,95){\makebox(0,0){$\cdots$}}
\put(60,60){\framebox(30,30){$\Gr{\pi_2}$}}
\putline{60}{111.6}{-1}{1.5}{3}
\putline{70}{111.6}{1}{1.5}{3}
\putmdelta{65}{107}
\puttext{65}{100}{$\bullet$}
\put(65,90){\line(0,1){10}}
\putline{80}{111.6}{-1}{1.5}{3}
\putline{90}{111.6}{1}{1.5}{3}
\putmdelta{85}{107}
\puttext{85}{100}{$\bullet$}

\put(85,90){\line(0,1){10}}
\put(75,95){\makebox(0,0){$\cdots$}}
\putcircle{45}{30}{20}{${\land}\I$}
\put(45,20){\line(0,-1){10}}
\puttext{45}{19}{$\bullet$}
\putdelta{45}{0}
\puttext{45}{10}{$\bullet$}
\putline{37.5}{-5.8}{-1}{-1}{10}
\putline{52.5}{-5.8}{1}{-1}{10}
\put(15,60){\line(1,-1){22.5}}
\put(75,60){\line(-1,-1){22.5}}
\puttext{115}{50}{$\Lra$}
\end{mnet}
\begin{mnet}{140}{105}
\put(0,60){\framebox(30,30){$\Gr{\pi_1}$}}
\putmdelta{5}{107}
\puttext{5}{100}{$\bullet$}
\putline{0}{111.6}{-1}{1.5}{3}
\putline{10}{111.6}{1}{1.5}{3}
\put(5,90){\line(0,1){10}}
\putmdelta{25}{107}
\puttext{25}{100}{$\bullet$}
\putline{20}{111.6}{-1}{1.5}{3}
\putline{30}{111.6}{1}{1.5}{3}
\put(25,90){\line(0,1){10}}
\put(15,95){\makebox(0,0){$\cdots$}}
\put(60,60){\framebox(30,30){$\Gr{\pi_2}$}}
\putline{60}{111.6}{-1}{1.5}{3}
\putline{70}{111.6}{1}{1.5}{3}
\putmdelta{65}{107}
\puttext{65}{100}{$\bullet$}
\put(65,90){\line(0,1){10}}
\putline{80}{111.6}{-1}{1.5}{3}
\putline{90}{111.6}{1}{1.5}{3}
\putmdelta{85}{107}
\puttext{85}{100}{$\bullet$}
\put(85,90){\line(0,1){10}}
\put(75,95){\makebox(0,0){$\cdots$}}
\put(15,60){\line(0,-1){10}}
\puttext{15}{50}{$\bullet$}
\putdelta{15}{40}
\putcircle{15}{0}{20}{${\land}\I$}
\put(22.4,7.5){\line(1.8,1){45}}
\qbezier(7.6,7.5)(0,18)(7.6,33)
\put(15,-10){\line(0,-1){10}}
\put(75,60){\line(0,-1){10}}
\puttext{75}{50}{$\bullet$}
\putdelta{75}{40}
\putcircle{75}{0}{20}{${\land}\I$}
\put(67.6,7.5){\line(-1.8,1){45}}
\qbezier(82.4,7.5)(90,18)(82.4,33)
\put(75,-10){\line(0,-1){10}}
\puttext{115}{50}{$\Lra^*$}
\puttext{115}{38}{${\small\mathsf{(I.H.)}}$}
\puttext{15}{-10}{$\bullet$}
\puttext{75}{-10}{$\bullet$}
\end{mnet}
\begin{mnet}{130}{105}
\put(0,60){\framebox(30,30){$\Gr{\pi_1}$}}
\put(5,90){\line(0,1){10}}
\put(15,95){\makebox(0,0){$\cdots$}}
\put(25,90){\line(0,1){10}}
\put(35,60){\framebox(30,30){$\Gr{\pi_1}$}}
\put(40,90){\line(0,1){10}}
\put(50,95){\makebox(0,0){$\cdots$}}
\put(60,90){\line(0,1){10}}
\put(42.6,17.5){\line(-1,1.6){26.5}}
\put(57.4,17.5){\line(1,1.6){26.5}}
\putcircle{50}{10}{20}{${\land}\I$}
\put(50,0){\line(0,-1){10}}

\put(70,60){\framebox(30,30){$\Gr{\pi_2}$}}
\put(75,90){\line(0,1){10}}
\put(85,95){\makebox(0,0){$\cdots$}}
\put(95,90){\line(0,1){10}}
\put(105,60){\framebox(30,30){$\Gr{\pi_2}$}}
\put(110,90){\line(0,1){10}}
\put(120,95){\makebox(0,0){$\cdots$}}
\put(130,90){\line(0,1){10}}
\put(92.4,17.5){\line(1,1.6){26.5}}
\put(77.6,17.5){\line(-1,1.6){26.5}}
\putcircle{85}{10}{20}{${\land}\I$}
\put(85,0){\line(0,-1){10}}
\puttext{50}{0}{$\bullet$}
\puttext{85}{0}{$\bullet$}

\end{mnet}
\end{center}
\end{itemize}
\qed
\end{proof}

\begin{proposition}
(Correctness) For each normalisation step transforming $\pi$ to
$\pi'$, there is a transformation from $\Gr{\pi} \Lra^* \Gr{\pi'}$.
\end{proposition}
\begin{proof} For the first normalisation rule, $\pi$ is a proof build from $\pi_1$ and $\pi_2$ using a $({\land}\I)$ followed by $({\land}\E_1)$ (similar for $({\land}\E_2)$), and $\Gr{\pi'}$ is a graph consisting of $\Gr{\pi_1}$ and a $W$ node attached to each hypothesis in $\Delta$ (note that $\Delta$ are the hypothesis for $\pi_2$). Using the normalisation step in Section \ref{sec:graphnorm}, from $\Gr{\pi}$ one obtains a graph consisting of $\Gr{\pi_1}$ and a $W$ node attached to the conclusion port of $\Gr{\pi_2}$. By induction on $\pi_2$, and relying on Lemma \ref{lem:erasing}, one can proof that a graph consisting to a $W$ node attached to the conclusion port of $\Gr{\pi_2}$, reduces to a graph consisting of a $W$ node attached to each hypothesis in $\Delta$.
For the second normalisation rule, we rely on Lemma
\ref{lem:duplication} to prove a similar result for $C$.
\qed
\end{proof}
\section{Graphs for the Linear \lam}
\label{sec:graphs-llc}

Now that we have seen what graph reduction is for the logic, we
briefly look at it again, but from a different perspective.  There are
standard ways of representing the \lam as graphs, and the reduction
mechanism as a graph rewriting system. We restrict this section to the
linear case which is simpler, and generalises using the structural
rules of the previous section. The general form of the translation
generating a graph from a term $t$ is the following:

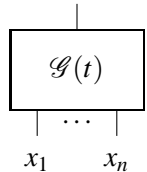
\begin{net}{50}{55}
\putbox{0}{20}{50}{30}{$\Gr{t}$}
\putVline{10}{10}{10}
\putVline{40}{10}{10}
\putVline{25}{50}{10}
\puttext{25}{15}{$\cdots$}
\puttext{10}{0}{$x_1$}
\puttext{40}{0}{$x_n$}
\end{net}
where the free variables of $t$ are $x_1,\ldots,x_n$.  The translation
$\Gr{\cdot}$ is given inductively over the structure of the linear
term $t$. We shall often drop the labelling of the edges when there is
no ambiguity. For abstractions we assume (without loss of generality)
that the (unique occurrence of the) variable $x$ occurs in the
leftmost position of the free variables of $\Gr{u}$. Notice also that,
in the case of applications there will not be any common free
variables between the graphs for $u$ and $v$ by the linearity
constraint.

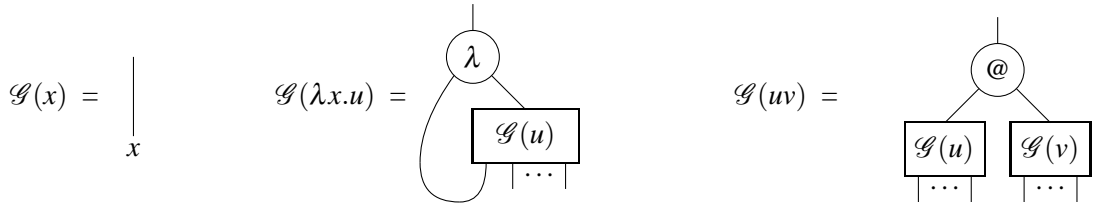
\begin{center}
\begin{mnet}{40}{65}
\puttext{0}{40}{$\Gr{x}$\ \ =}
\putVline{30}{25}{30}
\puttext{30}{20}{$x$}
\end{mnet}
\qquad
\qquad
\qquad
\begin{mnet}{100}{72}
\puttext{0}{40}{$\Gr{\lambda x.u}$\ \ =} 
\putbox{50}{15}{40}{20}{$\Gr{u}$}
\putVline{65}{5}{10}
\putVline{85}{5}{10}
\puttext{75}{10}{$\cdots$}
\putlam{50}{55}
\putVline{50}{65}{10}
\putline{70}{35}{-1}{1}{12.9}
\putcurve{(55,15)}{(55,0)}{(45,0)}{(30,0)}{(30,15)}
\qbezier(30,15)(30,30)(42.9,47.9)
\end{mnet}
\qquad
\qquad
\qquad
\begin{mnet}{130}{60}
\puttext{0}{40}{$\Gr{uv}$\ \ =}
\putVline{80}{60}{10}
\putbox{45}{10}{30}{20}{$\Gr{u}$}
\putbox{85}{10}{30}{20}{$\Gr{v}$}
\putapp{80}{50}
\putline{100}{30}{-1}{1}{12.9}
\putline{60}{30}{1}{1}{12.9}
\putVline{50}{0}{10}
\putVline{70}{0}{10}
\putVline{90}{0}{10}
\putVline{110}{0}{10}
\puttext{60}{5}{$\cdots$}
\puttext{100}{5}{$\cdots$}

\end{mnet}
\end{center}

As examples consider the graphs for: $(\lambda
x.x)(\lambda x.x)$ and $(\lambda xy.yx)$. 

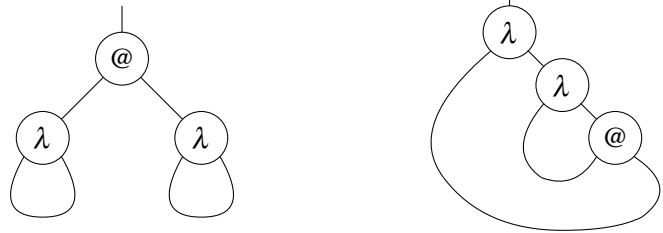
\begin{center}
\begin{mnet}{100}{73}
\putapp{50}{60}
\putVline{50}{70}{10}
\putlam{20}{30}
\putlam{80}{30}
\putline{27.1}{37.1}{1}{1}{15.8}
\putline{72.9}{37.1}{-1}{1}{15.8}
\putcurve{(27.1,22.9)}{(40,0)}{(20,0)}{(0,0)}{(12.9,22.9)}
\putcurve{(87.1,22.9)}{(100,0)}{(80,0)}{(60,0)}{(72.9,22.9)}
\end{mnet}
\qquad\qquad
\begin{mnet}{110}{87}
\putapp{90}{30}
\putlam{70}{50}
\putlam{50}{70}
\putVline{50}{80}{10}
\putline{62.9}{57.1}{-1}{1}{6}
\putline{82.9}{37.1}{-1}{1}{6}
\putcurve{(82.9,22,9)}{(75,10)}{(62,15)}{(48,25)}{(62.9,42.9)}

\qbezier(97.1,22.9)(115,11)(100,0)
\qbezier(100,0)(90,-5)(70,-5)
\qbezier(70,-5)(40,-5)(25,15)
\qbezier(25,15)(10,35)(42.9,62.9)
\end{mnet}
\end{center}

It is not difficult to see that replacing
${\limp}\I$\footnote{Simplified to the linear case.} by $\lambda$ and
${\limp}\E$ by $@$, and changing the orientation of these graphs we
obtain exactly the same system of graph reduction given for the logic.
This all leads to visual confirmation of the Curry-Howard isomorphism,
where we can think of graphs corresponding to proofs, types
corresponding to formulas, and graph reduction to normalisation.

We now turn to reduction in the linear \lam for these graphs, and set
up a notion of \emph{linear graph reduction}.  The idea is quite
simple: we will draw the graph for the term $(\lambda x.t)u$ and
another for $t[u/x]$ and try to deduce the corresponding graph
reduction step(s). The required reduction is given by:

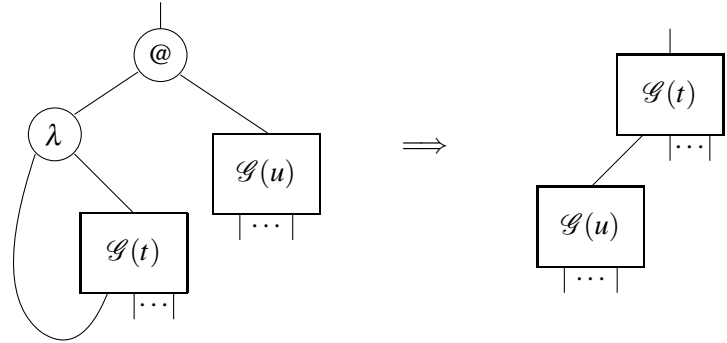
\begin{center}
\begin{picture}(280,130)
\put(30,20){\framebox(40,30){$\Gr{t}$}}
\put(50,20){\line(0,-1){10}}
\put(65,20){\line(0,-1){10}}
\put(57.5,15){\makebox(0,0){$\cdots$}}
\put(50,50){\line(-1,1){22.6}}
\putcircle{20}{80}{20}{$\lambda$}
\qbezier(40,20)(30,-5)(15,5)
\qbezier(15,5)(-5,20)(12.6,72.5)
\put(27.5,87.5){\line(3,2){24}}
\putcircle{60}{110}{20}{$@$}
\put(60,120){\line(0,1){10}}
\put(67.5,102.5){\line(3,-2){33}}
\put(80,50){\framebox(40,30){$\Gr{u}$}}
\put(90,50){\line(0,-1){10}}
\put(110,50){\line(0,-1){10}}
\put(100,45){\makebox(0,0){$\cdots$}}
\put(160,75){\makebox(0,0){$\Lra$}}
\put(200,10){
\begin{picture}(80,100)
\put(30,70){\framebox(40,30){$\Gr{t}$}}
\put(50,100){\line(0,1){10}}
\put(65,70){\line(0,-1){10}}
\put(50,70){\line(0,-1){10}}
\put(57.5,65){\makebox(0,0){$\cdots$}}
\put(0,20){\framebox(40,30){$\Gr{u}$}}
\put(10,20){\line(0,-1){10}}
\put(30,20){\line(0,-1){10}}
\put(20,15){\makebox(0,0){$\cdots$}}
\put(20,50){\line(1,1){20}}
\end{picture}}
\end{picture}
\end{center}

\begin{example}
Here are a couple of examples of linear graph reduction.  The first
example is the identity applied to the identity function. Now
$\Gr{(\lambda x.x)(\lambda x.x)}$ (which we call $N$) reduces to
$\Gr{\lambda x.x}$ (which we call $N'$) with one application of the
rewrite rule:
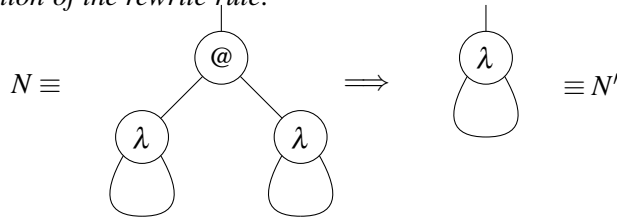
\begin{net}{170}{75}
\identity{0}{0}
\identity{60}{0}
\identity{130}{30}
\putapp{50}{60}
\putVline{50}{70}{10}
\putline{27.1}{37.1}{1}{1}{15.8}
\putline{72.9}{37.1}{-1}{1}{15.8}
\puttext{105}{50}{$\Lra$}
\putVline{150}{70}{10}
\puttext{-20}{50}{$N\equiv$}
\puttext{190}{50}{$\equiv N'$}
\end{net}
Note that there was only one graph reduction step required here: the
$\beta$-reduction step, together with the substitution, was captured
in a single rewrite. The advantage of this particular system of graph
rewriting for the \lam is that substitution is always done for free.
As a larger example, consider the term $(\lambda xy.yx)(\lambda x.x)(\lambda x.x)$.
As a graph reduction, this simply becomes the following:
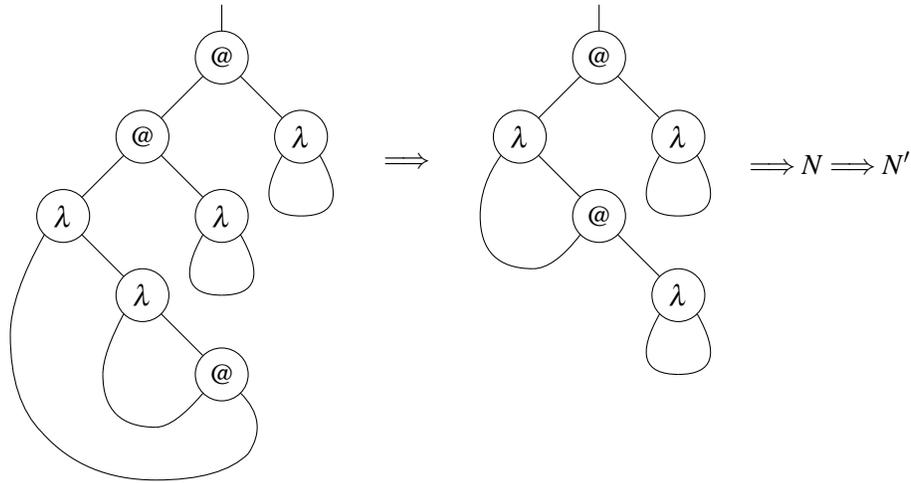
\begin{net}{350}{175}
\puttext{310}{120}{$\Lra N \Lra N'$}
\putapp{80}{160}
\putapp{50}{130}
\putapp{80}{40}
\putlam{20}{100}
\putlam{50}{70}
\identity{90}{100}
\identity{60}{70}
\putVline{80}{170}{10}
\putline{102.9}{137.1}{-1}{1}{15.8}
\putline{57.1}{137.1}{1}{1}{15.8}
\putline{27.1}{107.1}{1}{1}{15.8}
\putline{72.9}{107.1}{-1}{1}{15.8}
\putline{42.9}{77.1}{-1}{1}{15.8}
\putline{72.9}{47.1}{-1}{1}{15.8}
\qbezier(42.9,62.9)(35,48)(35,40)
\qbezier(35,40)(35,30)(40,25)
\qbezier(40,25)(45,20)(55,20)
\qbezier(55,20)(62,20)(72.9,32.9)
\qbezier(12.9,92.9)(0,70)(0,55)
\qbezier(0,55)(0,30)(12,18)
\qbezier(12,18)(28,0)(55,0)
\qbezier(55,0)(80,0)(90,10)
\qbezier(90,10)(98,22)(87.1,32.9)
\puttext{150}{120}{$\Lra$}
\put(170,0){
\begin{mnet}{100}{180}
\identity{60}{40}
\identity{60}{100}
\putapp{50}{160}
\putapp{50}{100}
\putlam{20}{130}
\putVline{50}{170}{10}
\putline{72.9}{137.1}{-1}{1}{15.8}
\putline{42.9}{107.1}{-1}{1}{15.8}
\putline{27.1}{137.1}{1}{1}{15.8}
\putline{72.9}{77.1}{-1}{1}{15.8}
\qbezier(12.9,122.9)(5,112)(5,100)
\qbezier(5,100)(5,80)(25,80)
\qbezier(25,80)(33,80)(42.9,92.9)
\end{mnet}
}
\end{net}
where $N$ and $N'$ are the nets of the previous example. Thus, the
total number of graph rewrite steps is just three, corresponding
exactly to the number of $\beta$-reductions performed at the level of
syntax.  One can see that there are a lot of benefits from a graphical
notation for this simple calculus: the graph rewriting process is
\emph{local}, meaning that at any time we only rewrite the part of the
graph connecting the application and abstraction; the rest of the
graph remains unchanged. 
\end{example}

\section{Discussion and Conclusion}\label{sec:disc}

Our goal was to provide a notation for intuitionistic logic (also for
the $\lambda$-calculus through the Curry-Howard isomorphism) that
shares some of the advantages of previous graphical notions such as
proof nets and interaction nets, but also simplifies and alleviates
some of the constraints.

\begin{itemize}
\item The graphical notation brings out the structure of the proof
visually, close to the abstract syntax, and consequently we believe it
to be quite natural.

\item This notation is preserved under normalisation (computation)
which means that we can animate the process. As part of this, we can
better understand the process of normalisation and substitution.

\item As the examples show, the diagrams also alleviate much of the
syntactic clutter which helps to bring out the structure of the
underlying proof.

\item We have established that normalisation
preserves the graphical notation, but we have assumed that proofs
always come from a natural deduction proof (i.e., through a
translation). We have deliberately avoided the question as to when an
arbitrary graph built from the nodes given is a valid proof
however. These questions are difficult to solve, and until we have
established the usefulness of the notation, we need not invest effort
into this. However, it remains a very important question that will
need to be addressed.

\item For the $\lambda$-calculus, our approach and motivation is
similar to that of \cite{VEX}. Our graphs are closer to the abstract
syntax trees, and we believe this is easier to relate to the syntax in
addition to allowing the process of substitution to be controlled
precisely.

\item Since the $\lambda$-calculus is the foundational calculus
underlying functional programming, this gives a starting point for a
visual approach for this paradigm.

\end{itemize}

\bibliographystyle{abbrv} 
\bibliography{bibfile}

\end{document}